\documentclass[11pt,a4paper]{article}
\usepackage{fullpage}
\usepackage{hyperref}
\usepackage{graphicx}

\usepackage{xcolor}

\usepackage{amsmath,amsthm,amsfonts,amssymb,thmtools}
\allowdisplaybreaks

\usepackage[capitalize]{cleveref}

\newtheorem{theorem}{Theorem}
\newtheorem{corollary}[theorem]{Corollary}
\newtheorem{proposition}[theorem]{Proposition}
\newtheorem{lemma}[theorem]{Lemma}
\newtheorem{definition}[theorem]{Definition}
\newtheorem{observation}[theorem]{Observation}
\newtheorem{assumption}[theorem]{Assumption}

\newcommand{\cond}{\mid}
\newcommand{\E}{\mathbb{E}}
\newcommand{\Tmar}{\hat{T}_{\mathrm{mix}}}

\title{The Careless Coupon Collector's Problem}
\author{Emilio Cruciani\thanks{European University of Rome, Italy.}\\\texttt{emilio.cruciani@unier.it} \and Aditi Dudeja\thanks{This work was initiated while the author was affiliated with the University of Salzburg. This project has received funding from the European Research Council (ERC) under the European Union’s Horizon 2020 research and innovation programme (grant agreement No 947702)}\\ \texttt{aditidudeja@cuhk.edu.cn}}
\date{}

\begin{document}

\maketitle

\begin{abstract}
We initiate the study of the Careless Coupon Collector's Problem (CCCP), a novel variation of the classical coupon collector, that we envision as a model for information systems such as web crawlers, dynamic caches, and fault-resilient networks.
In CCCP, a collector attempts to gather $n$ distinct coupon types by obtaining one coupon type uniformly at random in each discrete round, however the collector is \textit{careless}: at the end of each round, each collected coupon type is independently lost with probability $p$.
We analyze the number of rounds required to complete the collection as a function of $n$ and $p$. 
In particular, we show that it transitions from $\Theta(n \ln n)$ when $p = o\big(\frac{\ln n}{n^2}\big)$ up to $\Theta\big((\frac{np}{1-p})^n\big)$ when $p=\omega\big(\frac{1}{n}\big)$ in multiple distinct phases.
Interestingly, when $p=\frac{c}{n}$, the process remains in a metastable phase, where the fraction of collected coupon types is concentrated around $\frac{1}{1+c}$ with probability $1-o(1)$, for a time window of length $e^{\Theta(n)}$.
Finally, we give an algorithm that computes the expected completion time of CCCP in $O(n^2)$ time.
\end{abstract}

\section{Introduction}
The coupon collector's problem is a fundamental paradigm in probability theory and combinatorics, with a history dating back to 1708 as a problem formulated by Abraham de Moivre~\cite{ferrante2014coupon} and still studied today (see \cref{sec:related} for a discussion of related work).
In its classical form, a collector seeks to complete a collection of $n$ distinct coupons by receiving one coupon uniformly at random in each discrete round. 
The natural question is that of understanding the expected number of rounds required to complete the collection.
Beyond its theoretical elegance, the model has become a cornerstone in computer science, widely utilized for the analysis of randomized algorithms and data structures, providing an essential tool for understanding properties of various probabilistic processes~\cite{raghavan1995randomized,mitzenmacher2017probability}.

In this paper, we initiate the study of a novel variation that we call the Careless Coupon Collector's Problem (CCCP).
Unlike the traditional collector who retains every acquired item, the \textit{careless collector} may lose progress over time.
This model serves as a natural abstraction for modern dynamic information systems, such as \textit{web crawlers} maintaining fresh indices, \textit{distributed caches} with randomized expiration policies, and \textit{fault-resilient networks} where sensors may periodically disconnect. 
The modeling of these systems through CCCP is discussed in more detail in \cref{sec:applications}.

\paragraph*{Formal definition of CCCP}
Fix an integer $ n \ge 1 $, the number of distinct coupon types, and a parameter $ p = p(n) \in [0,1] $, a probability loss parameter.
Let $ [n] := \{1,2,\ldots,n\} $.
We define a discrete-time Markov chain $ (S_t)_{t\ge0} $ with state space $ 2^{[n]} $ modeling the set of coupons acquired by the careless collector in round $t$.
In particular, starting from $S_0 = \emptyset$, in each round $t\ge 0$ we define
\[
    S_{t+1} = (S_t \cup \{C_t\}) \setminus \{\, i : L_{i,t} = 1 \,\},
\]
where:
\begin{itemize}
\item $ C_t $ is a random element of $ [n] $ drawn independently and uniformly at random, modeling the new coupon acquired by the collector in round $t$;
\item $ \{L_{i,t}\}_{i\in[n]} $ are independent and identically distributed indicator random variables with parameter $p$, each modeling if the collector loses coupon $i$ in round $t$ (i.e., $L_{i,t}=1$) or keeps it (i.e., $L_{i,t}=0$).
\end{itemize}
Equivalently, at each step the collector first gains the coupon $C_t$,
then independently loses each coupon in their possession (including $C_t$) with probability $p$.
Note that obtaining a coupon $C_t \in S_t$ does not have any effect on the collection. 
Also note that the collector may lose multiple coupons in the same round, including the one they just obtained.
See \cref{fig:cccp fun} for an illustration of CCCP.

\begin{figure}[ht!]
    \centering
    \includegraphics[width=0.8\linewidth]{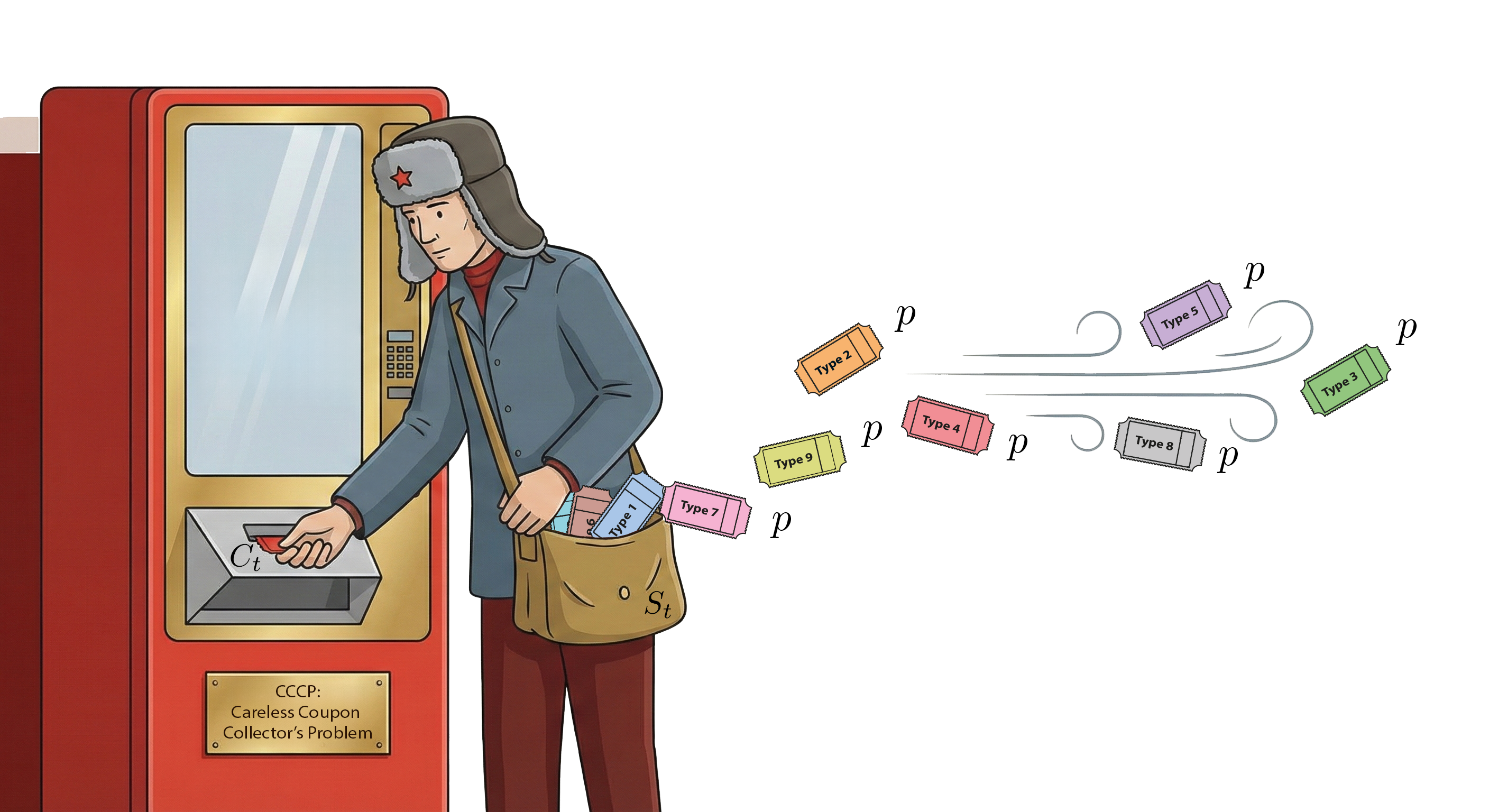}
    \caption{An illustration of a careless coupon collector.}
    \label{fig:cccp fun}
\end{figure}

We are interested in the first round $t$ such that the collector acquires all coupon types, namely at the following \textit{hitting time}
\[
    T_{n,p} := \min\{\, t \ge 0 : S_t = [n] \,\}.
\]
If $\{\, t \ge 0 : S_t = [n] \,\} = \emptyset$ the hitting time is undefined and we say $T_{n,p} = \infty$.
Our goal is to study the asymptotic behavior of 
$\E[T_{n,p}]$ as a function of $n$ and $p$.

\paragraph*{Behavior of the dynamics}
As discussed in detail in \cref{sec:someobservations}, the process generalizes the classical coupon collector's problem.
However, the loss of coupons makes CCCP significantly harder to analyze, since it breaks its monotonic growth and introduces ``catastrophic'' downward transitions (see e.g., \textit{skip-free Markov Chains} \cite{feller1971introduction} or \textit{birth-and-death processes with catastrophes}~\cite{brockwell1985extinction}), leading to a more complex behavior.

In particular, the loss parameter $p$ has two main effects on the dynamics.
The first is that of slowing down the completion of the collection. As $p$ increases, the expected hitting time undergoes a sharp qualitative change. 
Intuitively, for small $p$, the losses are negligible and the process behaves essentially like the classical coupon collector, completing quickly. As $p$ grows, losses start to matter and the hitting time increases rapidly, transitioning through distinct regimes. In particular, small variations of $p$ can trigger large, even exponential, increases in the hitting time, reflecting a phase-like shift from fast collection to a regime dominated by rare events (i.e., lucky streaks of acquisitions that allow for the completion of the collection).

The second effect of $p$ is to induce a metastable phase. Intuitively, the dynamics is governed by two competing forces: the acquisition of new coupons and the loss of already collected ones, both depending on $n,p$, and on the current number of collected coupons.
As the process evolves, these forces balance each other, creating an attractor of the dynamics that traps the system in a set of states where the fraction of collected coupons fluctuates around a limiting value $q_* = q_*(n,p)$.

In \cref{fig:simulation} we report some simulations illustrating these phenomena.
In particular, on the left we see the first effect of $p$ on the hitting time, that quickly transitions from being log-linear to exponential.
On the right, we observe how the fraction of collected coupons quickly stabilizes around the limiting fraction $q_*$. When $p=\Omega(1/n)$, the expected number of coupons $nq_*$ held at any given time becomes very small, dropping to a constant or even sub-constant value. In this regime, the natural state of the system is to be nearly empty. Consequently, completing the collection of all $n$ coupons requires the system to deviate drastically from its equilibrium. Reaching the full set essentially becomes a large stochastic fluctuation. In other words, the collector remains trapped in a metastable state where the number of coupons fluctuates around $nq_*$ for an exponential amount of time before a rare event allows to complete the collection.

\begin{figure}[!htbp]
    \centering
    \includegraphics[width=0.495\linewidth]{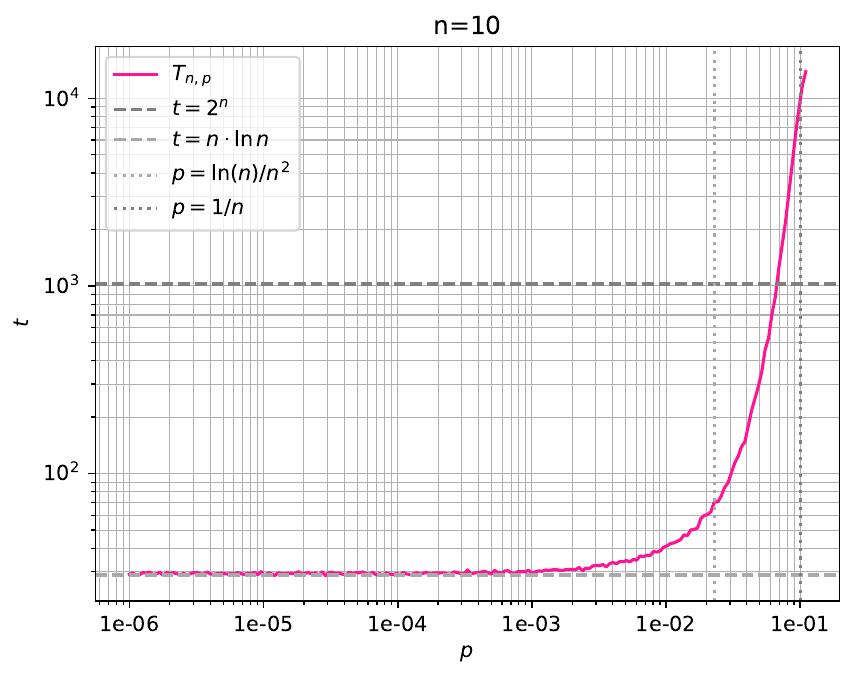}
    \includegraphics[width=0.495\linewidth]{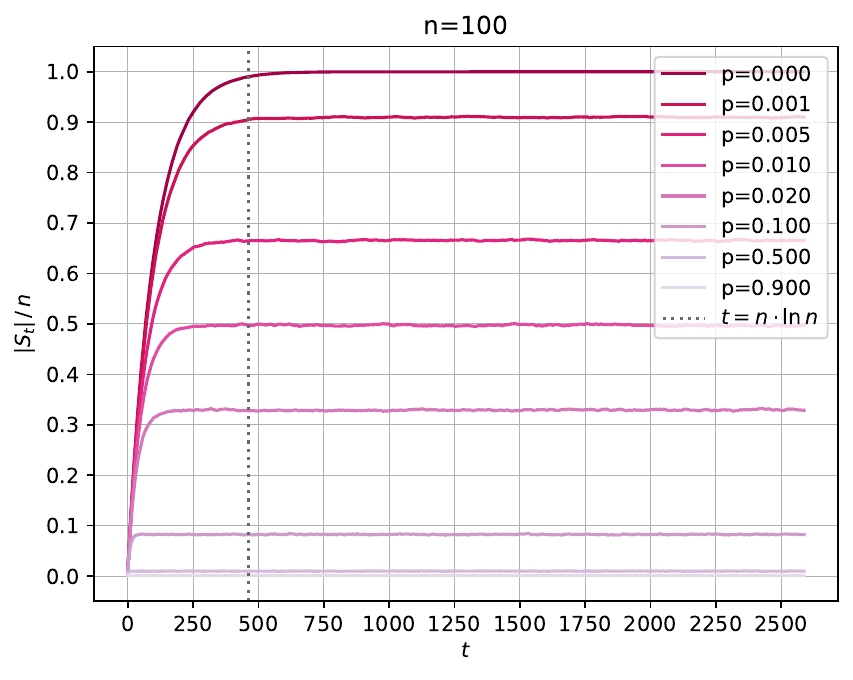}
    \caption{Empirical hitting time and metastable fractions of coupons of CCCP, averaged over 1000 repetitions.
    On the left: average hitting time while varying $p$, for fixed $n=10$; the small value of $n$ allows to wait exponential times for the process to end.
    On the right: average fraction $|S_t| / n$ of coupons in the collection at time $t$, for several values of $p$.}
    \label{fig:simulation}
\end{figure}

\paragraph*{Our contributions}
Our contributions can be summarized as follows:
\begin{itemize}
    \item \textbf{Model:} We introduce CCCP, a novel and natural generalization of the coupon collector's problem that, apart from the mathematical interest, can be used to model modern dynamic information systems (see \cref{sec:applications}).
    \item \textbf{Metastability:} We prove the existence of a \textit{metastable phase}, where the fraction of collected coupons remains concentrated for exponentially-many rounds with probability $1-o(1)$ around a limiting quantity $q_*$ that we exactly characterize as a function of $n,p$ (see \cref{sec:metastability}).
    \item \textbf{Hitting time:} We study the \textit{expected hitting time} at three different levels (see \cref{sec:hitting time}):
    \begin{itemize}
        \item We give an efficient \textit{exact algorithm} for the expected hitting time that runs in time $O(n^2)$ (\cref{sec:algorithm}).
        \item We perform a \textit{mean-field analysis} of the expected hitting time, through a spatiotemporal independence assumption: we give a simple asymptotic formulation that intuitively explains the hitting time of CCCP as a sum of two terms: the (marginal) \textit{mixing time} to the metastable state, and the \textit{escape time} from the metastable trap, waiting for the rare event in which the collection is completed (\cref{sec:mean field}). The expected hitting time described by this analysis essentially matches our simulations we presented earlier.
        \item We provide rigorous \textit{unconditional upper and lower bounds} to the expected hitting time (\cref{sec:bounds}). Our lower bound is asymptotically tight, matching the mean-field analysis.
    \end{itemize}
\end{itemize}

\paragraph*{Future research directions}
As CCCP is a newly proposed model, the primary significance of this work lies in identifying the collector's regimes and the discovery of the metastable phase. We believe this provides a foundation for future research directions:
\begin{itemize}
    \item The most immediate theoretical challenge is to provide a more refined upper bound to the expected hitting time that matches the lower bound, that we believe to be tight as it is asymptotically equal to the mean-field analysis.
    \item In this work, we assume independent and uniform losses. Future studies could investigate losses where each coupon type has its own probability of being loss $p_i$, on the same line as other classical models where the sampling of $C_t$ is not uniform, e.g., \cite{flajolet1992birthday}.
\end{itemize}

\subsection{Related work}\label{sec:related}

As already mentioned, the problem's origins date back to 1708, appearing as a challenge posed by De Moivre \cite{ferrante2014coupon}. 
However, the modern era of the coupon collector's problem began in the mid-20th century with rigorous results regarding hitting times and limit distributions established by von Schelling \cite{von1954coupon}, Newman \cite{newman1960double}, and the seminal work of Erdős and Rényi \cite{erdHos1961classical}.
The theoretical framework was later expanded by viewing the process through different lenses: 
Holst \cite{holst1986birthday} framed the collector's progress as a classic occupancy problem in urn models;
Flajolet et al. \cite{flajolet1992birthday} utilized the context of formal languages and generating functions to analyze birthday and coupon collection paradoxes.

Due to its elegant simplicity, the coupon collector's problem has found utility across diverse sectors.
In engineering it has been used to model reliability and optimal testing strategies \cite{luko2009coupon, dobson2015optimal}. 
In biology to model probabilistic structures in biological systems \cite{adler2005probabilistic}.
Moreover, it is a core concept in the analysis of randomized algorithms \cite{raghavan1995randomized, mitzenmacher2017probability} and has been the subject of extensive surveys regarding optimization \cite{boneh1997coupon}.

Recent research has shifted toward specialized variations that relax classical assumptions.
Under arbitrary sampling distributions, rather than uniform,
Berenbrink and Sauerwald \cite{berenbrink2009weighted} investigated approximations of the hitting times 
while Anceaume et al.~\cite{anceaume2015new, anceaume2016optimization} studied the time required to collect a given number of distinct coupons.
Raz and Zhan \cite{raz2020random} explored a memory-bounded variant where the collector must decide when to stop based on limited space, providing insights into the random-query model of computation.
Other studies have expanded the model to include the collection of multiple copies of each coupon \cite{adler2003coupon}, collecting coupons ``with friends''  \cite{alistarh2021collecting}, competing with other players \cite{krityakierne2025slowest}, through the best-of-the-sample paradigm where one samples a set and keeps the rarest items \cite{xu2011generalized}, and through the notion of super-coupons \cite{athreya2025generalisation}.
The problem continues to be generalized into highly abstract domains, such as quantum information and graph theory.
The introduction of the quantum coupon collector explores sampling in a quantum state framework \cite{arunachalam2020quantum}.
The labeled coupon collector has been applied to recovering bipartite graphs with perfect matchings \cite{berrebi2025labeled}.

Closer to the spirit of this paper, there exist also variants of the coupon collector that consider adversarial elements like resets, where the collector faces penalties or can lose all its coupons during the collection \cite{todic2023coupon, jockovic2024coupon}.

\subsection{Applications}\label{sec:applications}

This section highlights how analyzing key properties of CCCP, such as the expected hitting time and typical metastable states, can inform the design and analysis of dynamic information systems, including web crawlers, distributed caches, and fault-resilient networks. 

\paragraph*{Dynamic data under web crawling}
A web crawler monitors $n$ web pages that continuously change over time.
After visiting a page, its cached copy may become outdated if the page is modified.
We model this by assuming that each cached page becomes stale independently with probability $p$ at every step, a parameter driving the rate of content change.
The crawler revisits one page at a time, replacing stale copies with fresh ones.
The hitting time $T_{n,p}$ is the first moment when all $n$ cached pages are simultaneously up to date.
Only when $p$ is sufficiently small all pages align as fresh at once, corresponding to a complete, system-wide refresh of the index.
Most of the time, the crawler operates in a metastable state, where a stable fraction $q_*$ of pages remain fresh, while the rest lag behind.
The CCCP analysis can be used to estimate how frequently the crawler must revisit pages to maintain a desired level of global freshness, since higher revisit frequency corresponds to smaller $p$ and thus higher overall consistency.

\paragraph*{Cache expiration under load-balancing policies}
A distributed cache maintains $n$ items replicated across servers using a randomized expiration policy for load balancing.
To prevent cache hotspots or synchronization bursts, it periodically invalidate items at random, with each cached entry that expires independently with probability $p$ at every step.
When an item expires, it must be re-fetched from the origin server, similarly to re-collecting a lost coupon.
The hitting time $T_{n,p}$ represents the time until all $n$ items are simultaneously valid in the cache.
In practice, this full synchronization occurs rarely when $p$ is moderately large, while the system typically remains in a metastable regime where a stable fraction $q_*$ of cache entries are valid and the rest are pending refresh.
The CCCP analysis can help determine how frequently updates should occur to maintain high consistency across replicas, since more frequent updates correspond to smaller $p$ and thus shorter expected refresh times.

\paragraph*{Fault resilience in distributed networks}
Consider a network of $n$ sensors or devices that periodically report their status to a central server.
Each device may stop reporting with probability $p$ at each step, e.g., due to disconnection or failure, and the server polls one device at a time to re-establish its connection.
The system is fully synchronized only when all devices have recently reported, a rare event analogous to the hitting time $T_{n,p}$.
Most of the time, the system stays in a metastable regime where a steady fraction $q_*$ of devices are active and up to date. The CCCP analysis helps to determine how frequently the monitoring should occur to maintain high system observability and minimize missing reports.

\section{Preliminaries}\label{sec:prel}

We start the technical part by giving some simple observations regarding CCCP in \Cref{sec:someobservations}. Subsequently, in \Cref{sec:marginalprobability} we analyze the marginal probability of a coupon being present in the collected set of coupons at any time $t$. 
The marginal probability analysis will be our key tool in the remainder of the paper.

\subsection{Some observations on CCCP}
\label{sec:someobservations}

We discuss how $\E[T_{n,p}]$ varies depending on different values of $p$.
\begin{observation}\label{obs:classical ccp}
For $p = 0$, CCCP coincides with the classical coupon collector's problem, for which 
\[
    \E[T_{n,0}] = n \cdot H_n,
\]
where $H_n = \sum_{i=1}^{n} \frac{1}{i} = \ln n + O(1)$ denotes the $n$-th harmonic number (see, for example, \cite[Section 2.4.1]{mitzenmacher2017probability}).
Moreover, by standard probabilistic arguments, for any $c>0$ it holds that
\[
    \Pr(T_{n,0} > \E[T_{n,0}] + cn) = o(1).
\]
\end{observation}

It's obvious that $T_{n,p}$ is stochastically increasing in $p$, as having a higher chance of losing coupons cannot make the collection any faster. 
In particular, the following stochastic domination (that we indicate with $\preceq$) holds true.
\begin{proposition}\label{prop: monotone coupling}
Let $(S_t^{(p_1)})_{t \ge 0}$ and $(S_t^{(p_2)})_{t \ge 0}$ be CCCP chains with loss probabilities $p_1$ and $p_2$.
If $p_1 \le p_2$ then $T_{n,p_1} \preceq T_{n,p_2}$.
\end{proposition}
\begin{proof}
We prove the statement by a monotone coupling argument.
We use the same sampled coupon $C_t \in [n]$ for both chains at each time $t$.
For each coupon $i \in [n]$ and time $t \ge 0$, generate a uniform random variable $U_{i,t} \sim \text{Uniform}[0,1]$, and set
\[
	L_{i,t}^{(p_1)} = \mathbf{1}\{ U_{i,t} \le p_1 \}
	\qquad\text{and}\qquad
	L_{i,t}^{(p_2)} = \mathbf{1}\{ U_{i,t} \le p_2 \},
\]
where $\mathbf{1}\{E\}$ is the indicator random variable of the event $E$.
Since $p_1\le p_2$, then $L_{i,t}^{(p_1)} \le L_{i,t}^{(p_2)}$ for all $i,t$.
By induction on $t$, for all $t \ge 0$ we have that $S_t^{(p_1)} \supseteq S_t^{(p_2)}$.
Indeed, $S_0 = \emptyset$ in CCCP and adding the same drawn coupon $C_t$ and removing fewer lost coupons for $p_1$ preserves the inclusion.
Since $T_{n,p} := \min\{ t\ge 0 : S_t = [n] \}$, it follows immediately that $T_{n,p_1} \preceq T_{n,p_2}$.
\end{proof}

We now state some other observations about $T_{n,p}$ for $p>0$. 

\begin{observation}
For $p=1$, CCCP never ends since in each time step a coupon is collected and also lost with probability 1. Hence the process never reaches $[n]$ and we have
\[
    T_{n,1} = \infty.
\]
\end{observation}
\begin{observation}
For $p<1$, CCCP eventually completes the collection with probability 1, hence
\[
    T_{n,p} < \infty.
\]
This holds because, from any state $S$, there is a positive probability of reaching the target state $[n]$ in a finite number of steps.
Indeed there is a probability $\big(1-\frac{|S|}{n}\big)(1-p)^{|S|+1}$ of acquiring a new coupon without losing any existing ones, making a step forward towards $[n]$.
Since this probability is strictly positive $T_{n,p}$ is finite.
\end{observation}

\subsection{Marginal probabilities}
\label{sec:marginalprobability}

We analyze the probability of having a specific coupon in the collection at time $t$.
For each fixed coupon type $i$, let
\begin{equation}\label{eq:qt}
	q_t = q_{i,t} := \Pr(i \in S_t).
\end{equation}
Note that by symmetry $q_{i,t} = q_t$ for every $i$, hence we drop the $i$ in the notation.

\begin{lemma}\label{lemma: recurrence}
Let
\[
    a := (1-p)\Big(1-\frac{1}{n}\Big)
    \qquad\text{and}\qquad
    b := \frac{1-p}{n}.
\]
For every $t$, it holds that
\begin{equation}\label{eq: recurrence qt}
	q_{t+1} = a q_t + b.
\end{equation}
\end{lemma}
\begin{proof}
By the law of total probability, we have
\begin{align*}
q_{t+1} &= \Pr(i \in S_{t+1}) \\
&= \Pr(i \in S_{t+1} \land i \in S_t) + \Pr(i \in S_{t+1} \land i \not\in S_t) \\
&= \Pr(i \in S_{t+1} \cond i \in S_t) \cdot \Pr(i \in S_t)
    + \Pr(i \in S_{t+1} \cond i \not\in S_t) \cdot \Pr(i \not\in S_t) \\
&= (1-p) q_t + \frac{(1-p)}{n} (1-q_t).
\end{align*}
By definition $a := (1-p)\big(1-\frac{1}{n}\big)$ and $b := \frac{1-p}{n}$. Thus, we conclude that $q_{t+1} = aq_t + b$.
\end{proof}

\begin{lemma}\label{lemma: fixed point}
The recurrence in \cref{eq: recurrence qt} converges to a unique fixed point
\begin{equation}\label{eq: fixed point q*}
	q_* = \lim_{t \rightarrow \infty } q_t = \frac{1-p}{1-p+np}.
\end{equation}
\end{lemma}
\begin{proof}
Let $f(x)=ax+b$ be the mapping defining the recurrence in \cref{eq: recurrence qt}.
Since $|a|=(1-p)(1-\frac{1}{n})<1$, for every $x,y$ we have $|f(x)-f(y)| < |x-y|$, i.e., $f$ is a contraction mapping.
By the Banach fixed-point theorem, $f$ admits a unique fixed point $q_*$.
Consequently, $\lim_{t \rightarrow \infty } q_t = q_*$ regardless of the initial value $q_0$.
Solving the fixed-point equation $q_{*}=aq_{*}+b$ yields $q_* = \frac{b}{1-a} = \frac{1-p}{1-p+np}$.
\end{proof}

The convergence of $q_t$ to the unique fixed point $q_*$ can be understood as the system reaching a \textit{stochastic steady state}. This state represents the point where the two fundamental forces of CCCP, i.e., coupon acquisition and independent loss, reach an equilibrium. When the size of the collection is less than $nq_*$ the acquisition rate is higher than the cumulative loss rate, and viceversa. The fixed point $q_*$ thus serves as an attractor where the expected net change in the collection size is zero, defining a level around which the process will fluctuate for the majority of its duration, as we will see in the next section, before eventually hitting the completion state $[n]$.

\begin{corollary}\label{coro: metastable regimes}
Recall the definition of $q_*$ in \cref{eq: fixed point q*}.
It holds that
\[
q_* = \begin{cases}
    1-\Theta(np) & \text{if } p=o\big(\frac{1}{n}\big),
    \\
    \frac{1}{1+c} - o(1) & \text{if } p=\frac{c}{n}, \text{ for constant $c>0$},
    \\
    \Theta\big(\frac{1}{np}\big) & \text{if } p=\omega\big(\frac{1}{n}\big).
\end{cases}
\]
\end{corollary}
\begin{proof} 
We show the proof for each of the regimes.
\begin{enumerate}
\item Consider the case when $p=o(\frac{1}{n})$. In this case, we use Taylor's expansion $\frac{1}{1+x}=1-x+O(x^2)$ where $x=\frac{1}{1+(np-p)}$. Using routine algebra, we can simplify to get the desired bound. 
\item When $p=\frac{c}{n}$, we can simplify the expression for $q_*$ to obtain, $q_{*}=\frac{1-p}{1+c-p}$. This can be further rewritten as $q_*=\frac{1}{1+c}-o(1)$.
\item Consider the case when $p=\omega(\frac{1}{n})$. Then,
$ q_*=(\frac{1}{np})\cdot \frac{1-p}{1+\frac{1-p}{np}}$. In this regime, we can rewrite $np=o(1)$. Thus, $q_*=\Theta(\frac{1}{np})$.
\qedhere
\end{enumerate}
\end{proof}

\begin{lemma}\label{lemma: monotonicity}
Consider the recurrence in \cref{eq: recurrence qt} with initial condition $q_0$ and fixed point $q_*$ as defined in \cref{eq: fixed point q*}.
It holds that
\[
    q_{t} - q_* = a^t (q_0 - q_*).
\]
As a corollary:
\begin{itemize}
    \item If $q_0 < q_*$, then $q_t$ is monotonically increasing.
    \item If $q_0 > q_*$, then $q_t$ is monotonically decreasing.
\end{itemize}
\end{lemma}
\begin{proof}
Since $q_*$ is a fixed point, it holds that $q_* = a q_* + b$.
Therefore
\begin{align*}
q_{t} - q_* &=a q_{t-1} + b - (a q_* + b) \\
&= a (q_{t-1} - q_*) \\
&\; \; \vdots \\
&= a^t (q_0 - q_*).
\qedhere
\end{align*}
\end{proof}

\cref{lemma: fixed point} states that $q_t$ converges to $q_*$ and \cref{lemma: monotonicity} provides its convergence rate.
We also define the \textit{marginal mixing time} as the first time such that $q_t$ is sufficiently close to $q_*$.

\begin{definition}[Marginal mixing time]\label{def: T_mix marginal}
For any $\epsilon \in (0,1)$, let
\begin{equation}\label{eq: T_mix marginal}
    \Tmar
    = \Tmar(\epsilon) 
    := \min\left\{\, t\ge 0: \forall s\ge t,\ \big|q_s -q_*\big|\le\epsilon \,\right\}.
\end{equation}
\end{definition}
Note that the marginal mixing time differs from the standard notion of mixing time \cite{levin2017markov} since it only measures convergence of the marginals rather than of the entire state distribution.

\begin{lemma}\label{lem:boundonTmix}
For any $\epsilon \in (0,1)$, it holds that
\[
    \Tmar(\epsilon)
    = \frac{\ln(q_*/\epsilon)}{\ln(1/a)} 
    = \Theta\left(\frac{\ln(q_*/\epsilon)}{p+1/n}\right).
\]
\end{lemma}
\begin{proof}
By \cref{lemma: monotonicity}, for any arbitrary small $\epsilon > 0$, we have that
\[
    |q_t - q_*| = a^t q_* \le \epsilon 
    \iff t \ge \frac{\ln(q_*/\epsilon)}{\ln(1/a)} = \Tmar.
\]
Using the Taylor expansion $\ln\big(\frac{1}{1-x}\big) = x(1+o(1))$, since $a := (1-p)(1-1/n)$ we have
$\ln(1/a) = \left(p+\frac{1}{n}\right)(1+o(1))$. 
Therefore, we conclude that
$\Tmar(\epsilon) = \Theta\big(\frac{\ln(q_*/\epsilon)}{p+1/n}\big)$.
\end{proof}

\section{Metastability}\label{sec:metastability}
For every time $t\ge 0$ and coupon $i\in[n]$, let 
\begin{equation}\label{eq:X_it}
    X_{i,t} := \mathbf{1}\{\,i\in S_t\,\},
\end{equation}
namely $X_{i,t}$ is an indicator random variable of the event $\{\,i \in S_t\,\}$ that is 1 with probability $q_t$, by \cref{eq:qt}.
Hence, it holds that $\E[X_{i,t}] = q_t$.

\begin{lemma}\label{lemma: expectation St}
It holds that
\[
    \E[|S_t|] = n q_t.
\]
\end{lemma}
\begin{proof}
For each $t\ge 0$, by linearity of expectation we get
\begin{equation*}
    \E[|S_t|] 
    = \E\left[ \sum_{i=1}^n X_{i,t} \right]
    = \sum_{i=1}^n \E[X_{i,t}]
    = n q_t.
    \qedhere
\end{equation*}
\end{proof}

For each fixed time $t$, the random variables $(X_{i,t})_{i\in[n]}$ are not independent, however they are \textit{negatively associated}, as we prove next.

\begin{definition}[Negative association]\label{def:neg ass rvs}
The random variables $X_1,\dots,X_n$ are said to be negatively associated if for all disjoint subsets $I,J \subseteq \{1,\dots,n\}$ and all non-decreasing functions $f$ and $g$ it holds that
\[
    \E[f(X_i, i\in I) \cdot g(X_j, j\in J)] 
    \le \E[f(X_i, i\in I)] \cdot \E[g(X_j, j\in J)].
\]
\end{definition}

\begin{lemma}\label{lem:X_it neg ass}
For every fixed $t$, it holds that $(X_{i,t})_{i\in[n]}$ are negatively associated.
\end{lemma}
\begin{proof}
Recall that $S_0 = \emptyset$ and $S_t$ depends only on $S_{t-1}$, on the sampled coupon $C_{t-1}$ and indicators $\{L_{i,t-1}\}_{i \in [n]}$.
In particular, whether coupon $i$ belongs to $S_t$ depends only on coupons drawn at times $s\in\{0,\dots,t-1\}$ and whether those tokens survived the sequence of losses after up to time $t-1$.
Formally, for each time $s\in\{0,\dots,t-1\}$ and coupon $i\in[n]$ define
\[
	Y_{i,s} := Z_{i,s} \cdot R_{i,s},
\]
where 
\[
    Z_{i,s} := \mathbf{1}\{C_s=i\}
    \qquad\text{and}\qquad
    R_{i,s} := \prod_{r=s}^{t-1}\mathbf{1}\{L_{i,r}=0\}
\]
respectively indicate that coupon $i$ was sampled at time $s$, and that specific token survived every subsequent loss event up to time $t-1$.
With this definition, it holds that
\[
	X_{i,t} = \mathbf{1}\{\,i\in S_t\,\}
	= \bigvee_{s=0}^{t-1} Y_{i,s},
\]
namely coupon $i$ is in $S_t$ if and only if there exists at least one draw before time $t$ that sampled coupon $i$ and coupon $i$ survived all later losses up to $t-1$.

We prove negative association by starting from a family of negatively associated random variables and leveraging closure properties of negatively associated random variables \cite[Chapter~3.1]{dubhashi2009concentration}.
For each fixed $s$, the family of indicators $Z_{i,s}$ are such that exactly one variable is 1 and the rest are 0 and therefore is negatively associated.
For each fixed $s$, the $R_{i,s}$ are independent across $i$ and independent of the $Z_{i,s}$.
Combining a negatively associated family with an independent family preserves negative association.
Moreover the families $(Z_{i,s}, R_{i,s})$ are independent across different $s$.
Similarly, concatenating two independent families, each negatively associated, preserves negative association.
Finally, each $X_{i,t}$ is a coordinate-wise non-decreasing function of a disjoint block of $Y_{i,s}$'s.
Therefore, by the monotone aggregation property, the family $(X_{i,t})_{i}$ is negatively associated.
\end{proof}

Given the negative dependence, we can apply the following version of the Chernoff bound to prove that the expected fraction of collected coupons stays concentrated around $q_*$ for an exponential number of rounds.

\begin{theorem}[Multiplicative Chernoff bound, \cite{dubhashi2009concentration}]\label{thm:chernoff bound}
Let $X_1,\dots,X_n$ be independent or negatively associated indicator random variables. 
Let $X := \sum_{i=1}^{n} X_i$ and let $\mu = \E[X]$.
For all $\delta \in (0,1)$ it holds that:
\begin{equation*}
    \Pr(|X - \mu| > \delta\mu) \le 2e^{-\frac{\delta^2}{3}\mu}.
\end{equation*}
\end{theorem}

In particular, with the following proposition we prove the metastable regime for small values, namely when $p=O(1/n)$.

\begin{proposition}[Metastable number of coupons]\label{prop:metastability}
For any $\delta \in (0,1)$ pick a positive $\epsilon \le \delta q_*$ and let $T = \Tmar(\epsilon)$. 
For every $L \in \mathbb{N}$, it holds that
\[
    \Pr\left(\exists t \in \{T,\ldots,T + L-1\} : \big|\, |S_t| - nq_* \,\big| > 2 \delta nq_* \right) 
    \le 2L e^{- \frac{\delta^2(1-\delta)}{3} nq_*}.
\]
As a corollary, if $p = O(1/n)$, there exists $L = e^{\Theta(n)}$ such that
\[
    \Pr\left(\exists t \in \{T,\ldots,T + L-1\} : \big|\, |S_t| - nq_* \,\big| > 2 \delta nq_* \right) = o(1).
\]
\end{proposition}
\begin{proof}
Note that, by triangle inequality, we have
\begin{align*}
    \big|\, |S_t| - nq_* \,\big|
    &= \big|\, |S_t| -nq_t + nq_t - nq_* \,\big|
    \\
    &\le \big|\, |S_t| - nq_t \,\big| + \big| nq_t - nq_* \big|
    \\
    &= \big|\, |S_t| - nq_t \,\big| + n \big| q_t - q_* \big|.
\end{align*}

We analyze the two terms in the previous bound separately. 
For the first term, note that, by \cref{def: T_mix marginal}, for every $t \ge T = \Tmar(\epsilon)$ we have that $|q_t - q_*| \le \epsilon$, which implies $q_t \ge q_* - \epsilon \ge (1-\delta) q_*$, since $\epsilon \le \delta q_*$ by assumption.
Hence, since $\E[|S_t|]=nq_t$ by \cref{lemma: expectation St}, for every fixed $t \ge T$ by \cref{thm:chernoff bound} we have
\begin{equation*}
    \Pr\left( \big|\, |S_t| - nq_t \,\big| > \delta nq_t \right)
    \le 2e^{-\frac{\delta^2}{3}nq_t}
    \le 2e^{-\frac{\delta^2(1-\delta)}{3}nq_*}.
\end{equation*}
For the second term, by \cref{def: T_mix marginal}, for every $t \ge T$ deterministically we have
\[
    n \big| q_t - q_* \big| 
    \le n \epsilon
    \le \delta n q_*.
\]
Therefore, for every fixed $t \ge T$ we have
\[
    \Pr( \big|\, |S_t| - nq_* \,\big| > 2\delta nq_*  )
    \le 2 e^{- \frac{\delta^2(1-\delta)}{3} nq_*}.
\]
By taking a union bound over all times $t \in \{T,\ldots,T+L-1\}$ we conclude that
\[
    \Pr( \exists t \in \{T,\ldots,T+L-1\} : 
    \big|\, |S_t| - nq_* \,\big| > 2\delta nq_* )
    \le 2L e^{- \frac{\delta^2(1-\delta)}{3} nq_*}.
\]

Regarding the corollary, by \cref{coro: metastable regimes}, whenever $p=O(1/n)$ we have $q_* = \Omega(1)$. 
Hence the corollary follows by picking, $L = e^{cn}$ for any $c < \delta^2(1-\delta)q_* / 3$.
\end{proof}

When the mean is small, or $p$ is large, we cannot get as strong guarantees as implied by the previous proposition. 
Instead, we show the following bound.

\begin{proposition}\label{prop:metastability large p}
Assume $p = \omega(1/n)$ and let $T = \Tmar(\epsilon)$.
For any $\delta \in (0,1)$ and every $L \in \mathbb{N}$, it holds that 
\[
    \Pr\left(\exists t \in \{T,\ldots,T + L-1\} : |S_t| > \delta n \right) \le L e^{-\frac{\delta^2}{4}n}.
\]
As a corollary, there exists $L = e^{\Theta(n)}$ such that
\[
    \Pr\left(\exists t \in \{T,\ldots,T + L-1\} : |S_t| > \delta n \right) = o(1).
\]
\end{proposition}

The previous proposition covers the remaining case $p=\omega(1/n)$.
We note that the absolute deviations from the mean in \cref{prop:metastability,prop:metastability large p} have comparable magnitude, as $q_* = \Theta(1)$ for $p=O(1/n)$.
To prove \cref{prop:metastability large p}, we use the following version of Chernoff.

\begin{theorem}[Additive Chernoff bound, \cite{dubhashi2009concentration}]\label{thm:chernoff bound additive}
Let $X_1,\dots,X_n$ be independent or negatively associated indicator random variables. 
Let $X := \sum_{i=1}^{n} X_i$ and let $\mu = \E[X]$.
For any $\Delta > 0$ it holds that:
\begin{equation*}
    \Pr( X - \mu > \Delta) \le e^{-\frac{2\Delta^2}{n}}.
\end{equation*}
\end{theorem}

\begin{proof}[Proof of \Cref{prop:metastability large p}]
By \cref{lemma: expectation St}, the fact that $q_t \le q_*$ for all $t$, and \cref{coro: metastable regimes}, since $p = \omega(1/n)$ it follows that $\E[|S_t|] = nq_t \le nq_* = o(n)$.
Therefore, for any constant $\delta \in (0,1)$, it holds that $\delta n - q_t n \ge \frac{\delta}{2}n$.
By \cref{thm:chernoff bound additive} we have
\[
    \Pr\left( |S_t| > \delta n \right) 
    = \Pr\left( |S_t| - q_t n > \delta n - q_t n \right) 
    \le e^{-\frac{(\delta n - q_t n)^2}{n}}
    = e^{-(\delta - q_t)^2 n}
    \le e^{-\frac{\delta^2}{4}n}.
\]
By taking a union bound over all times $t \in \{T,\ldots,T+L-1\}$ we conclude that
\[
    \Pr\left(\exists t \in \{T,\ldots,T + L-1\} : |S_t| > \delta n \right) 
    \le L e^{-\frac{\delta^2}{4}n}.
\]

The corollary follows by picking $L = e^{cn}$, for any $c < \delta^2/4$.
\end{proof}

\section{Hitting time}\label{sec:hitting time}
In this section we study the expected hitting time of CCCP.
In particular in \cref{sec:algorithm} we give an algorithm to compute the exact hitting time, in \cref{sec:mean field} we perform a mean-field analysis of the expected hitting time, and in \cref{sec:bounds} we give upper and lower bounds that drop the mean-field assumptions.

\subsection{Exact computation}\label{sec:algorithm}
In this section we start to study the expected hitting time of CCCP. As we will see, the exact formulation does not admit a closed-form expression. 
However, we describe how it can be computed exactly, proving the following statement.
\begin{proposition}
There is an algorithm that, given input $n,p$, computes $\E[T_{n,p}]$ in time $O(n^2)$.
\end{proposition}
\begin{proof}
Let $K_t = |S_t|$. Note that in order to compute the expected hitting time $\E[T_{n,p}]$ we can study the simpler Markov Chain $(K_t)_{t\ge 0}$ associated to CCCP and note that
\[
    T_{n,p} = \min\{\, t\ge 0 : K_t = n \,\}.
\]
is an equivalent formulation of the hitting time of CCCP.
This process is simpler in that it consists only of $n+1$ states, i.e., $\{0,1,\ldots,n\}$.
In particular, the transition matrix $P \in \mathbb{R}^{(n+1) \times (n+1)}$ of the reduced process can be fully characterized as follows. 
For every $k$, we have:
\[
\begin{cases}
    P_{k,k-i} = \left(1-\frac{k}{n}\right)\binom{k+1}{i+1}(1-p)^{k-i}p^{i+1} + \frac{k}{n} \binom{k}{i}(1-p)^{k-i}p^i, & \forall i \in \{0,\ldots,k\},
    \\
    P_{k,k+1} = \left(1-\frac{k}{n}\right)(1-p)^{k+1},&\\
    P_{k,k+i} = 0, & \forall i \in \{2,\ldots,n-k\}.
\end{cases}
\]
Note that $P$ is stochastic (rows sums to 1) and has a special structure, being a \textit{lower-Hessenberg matrix}, i.e., all entries above the first superdiagonal are 0~\cite{golub2013matrix}. 
We give a graphical representation of CCCP and its corresponding transition matrix in \cref{fig:cccp-reduced mc + hessenberg}. Note that we will compute $\E[T_{n,p}]$ by solving a system of linear equations described by the transition matrix of CCCP. We mention the special structure of this matrix because it will be leveraged to speed-up the computation of $\E[T_{n,p}]$.

\begin{figure}[!ht]
    \centering
    \begin{minipage}{0.45\textwidth}
        \centering
        \includegraphics[width=\linewidth]{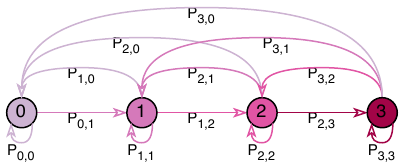}
    \end{minipage}
    \begin{minipage}{0.45\textwidth}
        \centering
        $\begin{pmatrix}
        P_{0,0} & {\color{magenta} P_{0,1}} & 0 & 0 \\
        P_{1,0} & P_{1,1} & {\color{magenta} P_{1,2}} & 0 \\
        P_{2,0} & P_{2,1} & P_{2,2} & {\color{magenta} P_{2,3}} \\
        P_{3,0} & P_{3,1} & P_{3,2} & P_{3,3} \\
        \end{pmatrix}$
    \end{minipage}
    \caption{A graphical representation of the reduced Markov Chain of CCCP with $n=3$ coupons (left) and its corresponding lower-Hessenberg transition matrix with first superdiagonal entries in magenta (right).}
    \label{fig:cccp-reduced mc + hessenberg}
\end{figure}

In order to do so, let us define now
\[
    h(k) := \E[T_{n,p} \cond K_0 = k].
\]
Since in CCCP $K_0 = 0$, by the law of total expectation we compute $h(0) = \E[T_{n,p} \cond K_0 = 0]= \E[T_{n,p}]$ recursively as follows:
\[
\begin{cases}
h(n) = 0,\\
h(k) = 1 + \sum_{i=0}^{k+1} P_{k,i} \cdot h(i), & \forall k \in \{0,\ldots,n-1\},
\end{cases}
\]
where $h(n)=0$ is the boundary condition that stops the recursion at the first time the process hits state $n$, and the sum in $h(k)$ goes up to $k+1$ since $P_{k,k+i}=0$ for all $i>1$.
In words, in this recursive formulation we pay 1 for each step of the chain plus the expected hitting time from the state we reach.
By plugging in the equivalence $h(n)=0$ we get a linear system of $n$ equations in $n$ unknowns $h(0),\dots,h(n-1)$.
Now, let $I$ be the identity matrix and $Q \in \mathbb{R}^{n \times n}$ be the matrix of transient states of CCCP (i.e., $P$ without the last row and column). 
We write the system in matrix form and, by rearranging terms, we get
\begin{equation}\label{eq:Ah=b}
    A \mathbf{h} = \mathbf{b},
\end{equation}
where $A = (I - Q) \in \mathbb{R}^{n \times n}$ is such that
\[
    A_{k,i} = \begin{cases}
        1 - P_{k,i} & \text{if } k=i,
        \\
        - P_{k,i} & \text{if } k\neq i,
    \end{cases}
\]
$\mathbf{b} = (1, \ldots, 1)^\intercal \in \mathbb{R}^n$ is the vector of all ones, and $\mathbf{h} = (h(0), \ldots, h(n-1))^\intercal$ is the vector of unknowns.
Note that the matrix $A$ is non-singular since $Q$ is the submatrix of transient states of CCCP and, therefore, the system in \cref{eq:Ah=b} has a unique solution that equals the vector of expected absorbing times.
Moreover, note that $A$ inherits the lower-Hessenberg structure of $P$. 
This structure allows to specialize the Gaussian elimination algorithm, solving the system more efficiently~\cite{golub2013matrix}.
We proceed as follows:
\begin{enumerate}
  \item \textbf{Forward-elimination:} The goal is to transform the system in \cref{eq:Ah=b} into an equivalent one with an upper-triangular coefficient matrix, i.e., all elements below the diagonal are 0.
    As in standard Gaussian elimination, for each pivot (a diagonal entry of $A$) we \textit{eliminate} the entries below the pivot and set them to 0 by subtracting the correctly rescaled pivot row, so that the system stays equivalent. 
    As $A$ is lower-Hessenberg, the first row has only two nonzero entries: when we subtract it from all rows below, only $O(n)$ terms are affected instead of $O(n^2)$.
    After eliminating the first column, also the second row has only two nonzero entries, and so on iteratively for every row.
    This property allows to perform the forward-elimination step efficiently.
    More formally, for each $c=0,\dots,n-2$, for each $r>c$, let $m_{r,c} := A_{r,c}/A_{c,c}$ and perform the following updates:
    \begin{align*}
        A_{r,c} &\gets A_{r,c} - m_{r,c} \cdot A_{c,c} = 0,
        \\
        A_{r,c+1} &\gets A_{r,c+1} - m_{r,c} \cdot A_{c,c+1},
        \\
        \mathbf{b}_{r} &\gets \mathbf{b}_{r} - m_{r,c} \cdot \mathbf{b}_{c}.
    \end{align*}
    For each $c$ there are $n-c-1$ entries below the diagonal that need to be eliminated, and each elimination step only updates two entries of $A$ and one entry of $\mathbf{b}$.
    Hence, summing over all columns, we have a total running time of $\sum_{c=0}^{n-2} 3(n-c-1) = O(n^2)$.

  \item \textbf{Backward-substitution:} After the forward-elimination, the system in \cref{eq:Ah=b} has been transformed into an equivalent upper-triangular system
    \[
        U \mathbf{h} = \mathbf{b}'.
    \]
    Note that $U$ is also still lower-Hessenberg, hence it has only two nonzero entries in each row, the diagonal and the first superdiagonal. 
    This also allows to speed-up the backward-substitution step computing the vector $\mathbf{h}$ as follows:
    \begin{align*}
        h(n-1) &= \frac{\mathbf{b}'_{n-1}}{U_{n-1,n-1}},
        \\
        & \,\,\,\vdots
        \\
        h(k) &= \frac{\mathbf{b}'_k - U_{k,k+1} \cdot h(k+1)}{U_{k,k}},
        \\
        & \,\,\,\vdots
        \\
        h(0) &= \frac{\mathbf{b}'_0 - U_{0,1} \cdot h(1)}{U_{0,0}}.
    \end{align*}
    Each step costs $O(1)$ to rearrange the equation, so backward-substitution costs $O(n)$ overall.
\end{enumerate}
The procedure computes the entire vector of absorbing states $\mathbf{h}$ and, in particular, the expected hitting time of CCCP $h(0)=\E[T_{n,p}]$ in time $O(n^2) + O(n) = O(n^2)$.
\end{proof}

\subsection{Mean-field analysis}\label{sec:mean field}

In this section we perform a mean-field analysis of the expected hitting time $T_{n,p}$. 
In particular we analyze it by assuming the spatial and temporal independence, as follows.
We remark that such assumptions do not hold in CCCP, but for the sake of the argument we make them nonetheless, as they are sufficient to derive clean bounds on the expected hitting time that capture the essential physics of CCCP.
Specifically, the analysis reveals a two-stage process: first, a rapid climb where the collection grows quickly toward the metastable state $nq_*$; second, a prolonged stochastic drift where the system fluctuates around $nq_*$ for an incredibly long duration.
In fact, interestingly, the bounds derived in this section essentially match the empirical values of the hitting time observed in the simulations described in the introduction (see \cref{fig:simulation}).

\begin{assumption}[Spatial and temporal independence]
Recall the random variables $(X_{i,t})_{i\in[n],t\ge 0}$ defined in \cref{eq:X_it}.
Assume:
\begin{itemize}\label{ass:spatial temporal independence}
    \item \textbf{Spatial independence:} For every fixed time $t$, the random variables $(X_{i,t})_{i\in[n]}$ are independent. Therefore, by \cref{eq:qt} we have
    \[
        \Pr(S_t = [n])
        = \Pr\left( \bigwedge_{i=1}^{n} \{\, X_{i,t}=1 \,\} \right)
        = q_t^n.
    \]
    \item \textbf{Temporal independence:} For every fixed coupon $i$, the process 
    \(
        (X_{i,t})_{t\ge 0}
    \)
    consists of independent random variables across time, and families corresponding to different coupons are mutually independent. Therefore, using also the spatial independence we have
    \[
        \Pr(T_{n,p} > t)
        = \Pr\left( \bigwedge_{j=0}^{t} \{\,S_j\neq[n]\,\} \right)
        = \prod_{j=0}^t \big(1 - q_j^{\,n}\big).
    \]
\end{itemize}
\end{assumption}

\begin{proposition}[Mean-field expected hitting time]\label{prop:mean-field}
Let $q_*$ as in \cref{eq: fixed point q*} and, for any $\epsilon>0$, let $\Tmar = \Tmar(\epsilon)$ as in \cref{def: T_mix marginal}. 
Under \cref{ass:spatial temporal independence}, it holds that:
\[
    \Tmar\cdot(1 - q_*^n\, \Tmar) + \frac{(1- q_*^n)^{\Tmar+1}}{q_*^n}
    \le \E[T_{n,p}]
    \le \Tmar + \frac{1}{(q_*-\epsilon)^n}.
\]
\end{proposition}
\begin{proof}
For the sake of readability, let $T=\Tmar(\epsilon)$.
Since $T_{n,p} \ge 0$, we can write the expected hitting time using the tail-sum formula as follows:
\begin{align}
    \E[T_{n,p}] &= \sum_{t=0}^{\infty} \Pr(T_{n,p} > t) \notag\\
    &= \sum_{t=0}^{T-1} \Pr(T_{n,p} > t) + \sum_{t=T}^{\infty} \Pr(T_{n,p} > t) \notag\\
    &= \sum_{t=0}^{T-1} \Pr(T_{n,p} > t) + \sum_{k=0}^{\infty} \Pr(T_{n,p} > T+k).
    \label{eq: E[T_np] split}
\end{align}
By our choice of $T$, we can consider the contributions of two terms in \cref{eq: E[T_np] split}: first, the \textit{marginal mixing time} where the $q_t$'s are still small and growing towards $q_*$ (see \cref{lemma: monotonicity}); 
second, the \textit{escape time} where the $q_t$'s are very close to $q_*$ and the hitting time can be approximated by a geometric random variable thanks to the independence assumptions.

In the first term of \cref{eq: E[T_np] split}, for $t < T$, we have that $q_t < q_*$. 
By using that $\Pr(T_{n,p} > t) \le 1$ we have
\begin{equation*}
    \sum_{t=0}^{T-1} \Pr(T_{n,p} > t) \le T.
\end{equation*}
By a union bound and due to the spatial independence assumption, we get
\begin{align*}
    \sum_{t=0}^{T-1} \Pr(T_{n,p} > t) 
    &= \sum_{t=0}^{T-1} \left( 
        1 - \Pr\bigg(  \bigvee_{j=0}^{t} \{\, S_j = [n] \,\} \bigg) 
    \right) 
    \\
    &\ge \sum_{t=0}^{T-1} \left(
        1 - \sum_{j=0}^t q_j^n
    \right)
    \\
    &\ge T (1 - T q_*^n).
\end{align*}

In the second term of \cref{eq: E[T_np] split}, for $t \ge T$, we have that $|q_t - q_*| \le \epsilon$. 
Therefore, as a consequence of the spatial and temporal independence assumptions we get
\begin{align*}
    \sum_{k=0}^{\infty} \Pr(T_{n,p} > T+k)
    &= \sum_{k=0}^{\infty} \prod_{j=0}^{T+k} \left( 1 - q_j^n \right)
    \\
    &\le \sum_{k=0}^{\infty} \prod_{j=T}^{T+k}
    \left( 1 - q_j^n \right)
    \\
    &\le \sum_{k=0}^{\infty} (1- (q_* - \epsilon)^n)^k
    \\
    &= \frac{1}{(q_*-\epsilon)^n}.
\end{align*}
Similarly, since $q_0 = 0$ by \cref{lemma: monotonicity} we have $q_j \le q_*$ for all $j$, and therefore
\begin{align*}
    \sum_{k=0}^{\infty} \Pr(T_{n,p} > T+k)
    &= \sum_{k=0}^{\infty} \prod_{j=0}^{T+k} \left( 1 - q_j^n \right)
    \\
    &\ge \sum_{k=0}^{\infty} (1- q_*^n)^{T+k+1}
    \\
    &= (1- q_*^n)^{T+1} \sum_{k=0}^{\infty} (1- q_*^n)^{k}
    \\
    &= \frac{(1- q_*^n)^{T+1}}{q_*^n}.
\end{align*}
We can substitute the bounds we derived in \cref{eq: E[T_np] split} to conclude the proof. 
\end{proof}
\begin{observation}
Under \cref{ass:spatial temporal independence}, for sufficiently small positive $\epsilon$, we have
\[
    \E[T_{n,p}] \approx \Tmar + q_*^{-n}.
\]
\end{observation}
The previous approximation defines the following \textit{regimes} of the expected hitting time:
\[
\E[T_{n,p}] \approx \begin{cases}
    \Theta(n \ln n) & \text{if } p=o\big(\frac{\ln n}{n^2}\big), \\
    n^{\Theta(c)} & \text{if } p=c\,\frac{\ln n}{n^2}, \text{ for constant $c>0$}, \\
    e^{\Theta(n^2 p)} & \text{if } p=\omega\big(\frac{\ln n}{n^2}\big) \text{ and } p=o(1/n), \\
    \Theta\big((1+c)^n\big) & \text{if } p=\frac{c}{n}, \text{ for constant $c>0$}, \\
    \Theta\big((\frac{np}{1-p})^n\big) & \text{if } p=\omega\big(\frac{1}{n}\big).
\end{cases}
\]
We describe these regimes in more detail. Towards this, recall the bounds on $q_*$ (\Cref{lemma: fixed point}) and $\Tmar$ (\Cref{lem:boundonTmix}). 
\begin{enumerate}
\item \label{item:classical}\textbf{Classical regime:} When $p=o(\frac{\ln n}{n^2})$, then $\Tmar$ dominates. One can show that $\Tmar=\Theta(n\ln n)$. Thus, CCCP asymptotically behaves as if $p=0$, i.e., as the classical coupon collector's problem.
\item\label{item:superclassical} \textbf{Super-classical regime:} When $p=\frac{c\ln n}{n^2}$, then $q_*^{-n}$ dominates. One can show that $\E[T_{n,p}]$ is ``super-classical'' but polynomial. This is is because,
\begin{align*}
    q_{*}^{-n}=\left(\frac{1-\frac{c\ln n}{n^2}}{1-\frac{c\ln n}{n^2}+\frac{c\ln n}{n}}\right)^{-n}=\left(1-\frac{\frac{c\ln n}{n}}{1-\frac{c\ln n}{n^2}+\frac{c\ln n}{n}}\right)^{-n}\approx e^{\frac{c}{1-o(1)}\ln n}=n^{\Theta(c)}. 
\end{align*}
\item\label{item:metastable1} \textbf{Metastable regime I:} When $p=\omega(\frac{\ln n}{n^2})$ and $p=o(\frac{1}{n})$, then $q_*^{-n}$ dominates, which we can bound as follows.
\begin{align*}
    q_{*}^{-n}=\left(\frac{1-p}{1-p+np}\right)^{-n}\approx e^{\frac{n^2p}{1-p}}=e^{\Theta(n^2p)}.
\end{align*}
Thus, $\E[T_{n,p}]$ is super-polynomial but sub-exponential.
\item\label{item:metastable2}\textbf{Metastable regime II:}  When $p=\frac{c}{n}$, we can again see that $q_*^{-n}$ dominates. Using routine algebra, we have,
\begin{align*}
    q_*^{-n}=\left({1+\frac{c}{1-p}}\right)^{n}=\Theta((1+c)^n).
\end{align*}
Thus, $\E[T_{n,p}]$ is exponential.
\item\label{item:metastable3} \textbf{Metastable regime III:} When $p=\omega(\frac{1}{n})$, we can again see that $q_*^{-n}$ dominates. Using routine algebra, we have,
\begin{align*}
    q_*^{-n}=\left(\frac{1-p+np}{1-p}\right)^{n}=\Theta\bigg(\Big(\frac{np}{1-p}\Big)^n\bigg).
\end{align*}
Thus, $\E[T_{n,p}]$ is super-exponential.
\end{enumerate}
While even more regimes can be defined, we believe these are the most interesting ones: 
\ref{item:classical}~and~\ref{item:superclassical} describe the regime for which the hitting time of CCCP is still polynomial; on the other hand \ref{item:metastable1}, \ref{item:metastable2}, and \ref{item:metastable3} describe the hitting times in the phases where CCCP has the three distinct regimes of the metastable fraction of collected coupons $q_*$ (see \cref{coro: metastable regimes}).

Note that in regimes \ref{item:classical}, \ref{item:superclassical}, and \ref{item:metastable1} we have that $q_* = 1-o(1)$, therefore the metastable window of exponential length described in \cref{prop:metastability,prop:metastability large p} might include fluctuations hitting state $[n]$, not contradicting the sub-exponential hitting time stated above.

Finally, we observe that the concentration of the hitting time $T_{n,p}$ around its expectation varies significantly across these regimes. 
Recall that in the classical coupon collector's problem ($p=0$) the hitting time $T_{n,0}$ is concentrated around its expectation (see \cref{obs:classical ccp}). 
In the classical regime ($p=o\big(\frac{\ln n}{n^2}\big)$), using \cref{ass:spatial temporal independence}, it is possible to prove a similar concentration of the hitting time $T_{n,p}$.
Informally, this is because the probability that the collector completes the collection after the marginal mixing time $\Tmar$ is large enough as the fraction of collected coupon is $q_* = 1-\Theta(np)$.
In contrast, in the other regimes, such a concentration does not hold. Indeed, the expected hitting time is dominated by the escape time from the metastable state which, under our mean-field assumptions, is approximated by a geometric random variable. Since the standard deviation of a geometric distribution is of the same order as its mean, the hitting time $T_{n,p}$ in these regimes undergoes large fluctuations and is not concentrated around its expected value.

\subsection{Asymptotic bounds}\label{sec:bounds}
Finally, in this section we drop \cref{ass:spatial temporal independence} and provide rigorous unconditional upper and lower bounds on the expected hitting time of CCCP.
In particular, we first give an upper bound that is still derived by the definition of a geometric random variable but with a different parameter compared to that of the mean-field analysis of \cref{sec:mean field}.
We conclude the section with a lower bound that, instead, is asymptotically tight with respect to our mean-field bound.

Let us consider the reduced chain $(K_t)_{t\ge 0}$ with $K_t = |S_t|$.
For $\epsilon\in(0,1)$, we define
\begin{equation}\label{eq:good states}
    \mathcal{G} = \mathcal{G}(\epsilon) := \{\, k \in \{0,\dots,n\} : k \ge (1-\epsilon) n q_* \,\}
\end{equation}
as the good set of states, namely that set of states that we reach in time $\Tmar$ and from which CCCP hits $[n]$ by waiting for the rare event in which all missing coupons are collected, and the others are not lost. 

\begin{lemma}\label{lemma:rho}
For $k \in \{0,1,\ldots,n\}$ and $\mathcal{G}$ as in \cref{eq:good states}, it holds that
\begin{align*}
    \min_{k \in \mathcal{G}} \Pr( \exists s \in \{t+1,\ldots,t + n \ln n\} : K_t = k, K_s = n )
    \ge q_* (1-\epsilon) (1-p)^{n^2 \ln n} =: \rho.
\end{align*}
\end{lemma}
\begin{proof}
For every $m \ge 1$, let 
\[
    F(k,m) := \{\, \exists s \in \{t+1,\ldots,t+m\} : K_t = k, K_s = n \,\}.
\]
Let us define the following two events needed to satisfy $F(k,m)$:
\begin{itemize}
    \item $E_1$: ``all the remaining $n-k$ coupons appear at least once in the $m$ draws of the block'';
    \item $E_2$: ``all coupons survive every step in the block''.
\end{itemize}
Note that $\Pr(F(k,m)) \ge \Pr(E_1 \land E_2) = \Pr(E_1) \cdot \Pr(E_2)$ as $E_1,E_2$ are independent in CCCP (the draws $C_t$ are independent of the losses $L_{i,t}$).
We have that
\begin{align*}
    \Pr(E_1) &\ge 1 - (n-k)(1-1/n)^m \ge 1-(n-k)e^{-m/n},
    \\
    \Pr(E_2) &\ge (1-p)^{mn}.
\end{align*}
Note that the minimum is achieved for $k = \min(\mathcal{G}) \ge (1-\epsilon)nq_*$, hence
\begin{align*}
    \min_{k \in \mathcal{G}} \Pr( F(k,m) ) 
    &\ge (1 - (n-k)e^{-m/n} ) \cdot (1-p)^{mn}
    \\
    &\ge (1 - (1- q_*(1-\epsilon))ne^{-m/n} ) \cdot (1-p)^{mn}.
\end{align*}
The lower bound is (approximately) optimized by picking $m = n \ln n$. Hence, we get
\[
    \min_{k \in \mathcal{G}} \Pr( F(k,m) ) 
    \ge q_* (1-\epsilon) (1-p)^{n^2 \ln n}.
    \qedhere
\]
\end{proof}

\begin{proposition}\label{prop:unconditional upper bound E[Tnp]}
Recall the definition of $\rho$ in \cref{lemma:rho}.
It holds that
\[
    \E[T_{n,p}] \le \Tmar(\epsilon) + \frac{n \ln n}{\rho}.
\]
\end{proposition}
\begin{proof}
Recall the definition of $\mathcal{G}$ in \cref{eq:good states} and define the bad event $E$ of exiting $\mathcal{G}$ before completing the collection as
\[
    E := \{\, \exists s, \Tmar(\epsilon) \le s < T_{n,p} : K_s \not\in \mathcal{G} \,\}.
\]
Let $\bar{E}$ be the complement of $E$.
We write the expectation of the hitting time as follows:
\begin{align*}
    \E[T_{n,p}] 
    &= \E[T_{n,p} \cond E] \cdot \Pr(E) + \E[T_{n,p} \cond \bar{E}] \cdot (1 - \Pr(E))
    \\
    &\le \Pr(E) \cdot \E[T_{n,p}] + (1-\Pr(E)) \cdot \E[T_{n,p} \cond \bar{E}],
\end{align*}
where $\E[T_{n,p} \cond E] \le \E[T_{n,p}]$ holds since the expected hitting time starting from any state not in $\mathcal{G}$ is largest when starting from state 0.
By rearranging the terms, we have
\[
    \E[T_{n,p}] - \Pr(E)\cdot\E[T_{n,p}] \le (1-\Pr(E)) \cdot \E[T_{n,p} \cond \bar{E}]
\]
and since $\Pr(E) < 1$ we conclude
\[
    \E[T_{n,p}] \le \E[T_{n,p} \cond \bar{E}].
\]

Let $m = n \ln n$. 
In the remainder of the proof we show that 
$\E[T_{n,p} \cond \bar{E}] \le \Tmar(\epsilon) + \frac{m}{\rho}$, for $\rho$ as in \cref{lemma:rho}.
We partition the time after $T=\Tmar(\epsilon)$ into disjoint blocks of length $m$, where each block is
\[
    B_i := \{T+im,\ldots,T+(i+1)m-1\}, \forall i\ge 0.
\]
For any $r \ge 0$, define the event
\[
    F_r := \{\, T_{np} > T+rm-1 \,\},
\]
namely ``no block $B_0,...,B_{r-1}$ contains a time $t$ where $K_t=n$''.
Note that
\[
    \Pr(F_r \cond \bar{E}) \le (1-\rho)^r.
\]

Recall \cref{eq: E[T_np] split} and write:
\begin{align*}
    \E[T_{n,p} \cond \bar{E}] 
    &= \sum_{t=0}^{T-1} \Pr(T_{n,p} > t \cond \bar{E}) + \sum_{t=T}^{\infty}\Pr(T_{n,p} > t \cond \bar{E})
    \\
    &= \sum_{t=0}^{T-1} \Pr(T_{n,p} > t \cond \bar{E}) + \sum_{r=0}^{\infty} \sum_{t \in B_r}\Pr(T_{n,p} > t \cond \bar{E}).
\end{align*}
By using $\Pr(T_{n,p} > t \cond \bar{E}) \le 1$ and noting that $\Pr(T_{n,p} > t)$ is decreasing in $t$ we get
\begin{align*}
    \E[T_{n,p} \cond \bar{E}] 
    &\le T + \sum_{r=0}^{\infty} m \Pr(T_{n,p} > T + rm - 1 \cond \bar{E})
    \\
    &= T + m\sum_{r=0}^{\infty} \Pr(F_r \cond \bar{E})
    \\
    &\le T + m\sum_{r=0}^{\infty} (1-\rho)^r
    \\
    &= T + \frac{m}{\rho}.
\end{align*}
We conclude the proof having $\E[T_{n,p}] \le \E[T_{n,p} \cond \bar{E}] \le \Tmar(\epsilon) + \frac{n \ln n}{\rho}$, for $\rho$ as in \cref{lemma:rho}.
\end{proof}


\begin{proposition}\label{prop:unconditional lower bound E[Tnp]}
It holds that $\E[T_{n,p}] \ge \frac{1}{5 q_*^n}$.

\end{proposition}
\begin{proof}
Recall the definition of $X_{i,t}$ in \cref{eq:X_it} and that they are negatively associated (\cref{lem:X_it neg ass}).
By using negative association and that $q_t \le q_*$ for all $t$ (by \cref{lemma: monotonicity}, since $q_0 = 0$), we have
\[
    \Pr(S_t = [n]) = \Pr\left( \bigwedge_{i=1}^n \{i \in S_t \}\right)
    \le \prod_{i=1}^n \Pr(i \in S_t)
    \le q_*^n.
\]
Moreover, by a union bound
\[
    \Pr(T_{n,p} \le t) = \Pr(\exists j \le t : S_j = [n])
    \le \sum_{j=0}^t \Pr(S_j = [n])
    \le (t+1) q_*^n
\]
and therefore
\[
    \Pr(T_{n,p} > t) \ge 1 - (t+1) q_*^n.
\]
Let $\bar{t} := \big\lfloor \frac{1 - 2 q_*^n}{2 q_*^n} \big\rfloor$ and note that for every $t \le \bar{t}$ we have
\[
    \Pr(T_{n,p} > t) \ge \frac{1}{2}.
\]
Using the tail-sum formula of the expectation and the previous observation we get
\[
    \E[T_{n,p}] = \sum_{t=0}^{\infty} \Pr(T_{n,p} > t) 
    \ge \sum_{t=0}^{\bar{t}} \Pr(T_{n,p} > t) 
    \ge \frac{\bar{t}}{2}
    \ge \frac{1}{5q_*^n}.
    \qedhere
\]
\end{proof}

We acknowledge that while our results correctly identify the exponential nature of the hitting time, there is a significant quantitative gap between the lower and upper bounds in the non-classical regimes. 
Our lower bound of $\approx 1/q_*^n$ is expected to be tight, matching the mean-field analysis. However, our rigorous upper bound is $\approx 1/(q_*(1-p)^{n^2 \ln n})$.
This gap stems from the conservative nature of our block-partitioning argument: to ensure a rigorous result without assuming temporal independence, we bound the probability of completing the set within a safe window where no losses occur. We conjecture that the true hitting time is much closer to the lower bound. 
Moving from this block-based approach to a more sophisticated analysis of the Markov chain is key for improving the bound.

\bibliographystyle{alpha}
\bibliography{references}

\end{document}